\documentclass[reqno,12pt]{amsart}
\usepackage{hyperref}
\usepackage{amssymb}
\usepackage{amsmath}
\usepackage{amsthm}
\usepackage[arrow,matrix,curve]{xy}
\usepackage{graphicx}
\usepackage{hyperref}
\usepackage{amscd,verbatim}

\usepackage{fullpage,enumerate}
\newcommand{\bea}{\begin{eqnarray}}
\newcommand{\eea}{\end{eqnarray}}

\newtheorem{theorem}{Theorem}[section]

\newtheorem{lemma}[theorem]{Lemma}

\theoremstyle{remark}
\newtheorem{definition}[theorem]{Definition}

\newtheorem{example}[theorem]{Example}

\numberwithin{equation}{section}
\numberwithin{theorem}{section}

\newcommand{\sfSp}{{\mathsf{Sp}}}
\newcommand{\sfO}{{\mathsf O}}
\newcommand{\sfSO}{{\mathsf{SO}}}
\newcommand{\sfSU}{{\mathsf{SU}}}

\newcommand\cA{\mathcal{A}}

\newcommand\bCP{\mathbb{C}\mathrm{P}}

\newcommand\bZ{\mathbb{Z}}
\newcommand\Diff{\operatorname{Diff}}

\newcommand\Isom{\operatorname{Isom}}

\newcommand{\longhookrightarrow}{\ensuremath{\lhook\joinrel\relbar\joinrel\rightarrow}}

\renewcommand{\H}{{\operatorname{H}}}
\def\K{ \hbox{\rm K}}
\newcommand{\Z}{\ensuremath{\mathbb Z}}

\newcommand{\RR}{{\mathbb R}}

\newcommand{\ZZ}{{\mathbb Z}}

\begin{document}

\title[Spherical T-duality II]
%{Non-principal $\sfSU(2)$-bundles: \\ an infinity of spherical T-duals}
{Spherical T-duality II: \\
An infinity of spherical T-duals for\\  non-principal $\sfSU(2)$-bundles}

\author[P Bouwknegt]{Peter Bouwknegt}

\address[Peter Bouwknegt]{
Mathematical Sciences Institute, and
Department of Theoretical Physics,
Research School of Physics and Engineering, 
The Australian National University, 
Canberra, ACT 0200, Australia}

\email{peter.bouwknegt@anu.edu.au}

\author[J. Evslin]{Jarah Evslin}

\address[Jarah Evslin]{
High Energy Nuclear Physics Group, 
Institute of Modern Physics,
Chinese Academy of Sciences,
Lanzhou, China}

\email{jarah@impcas.ac.cn}

\author[V Mathai]{Varghese Mathai}

\address[Varghese Mathai]{
Department of Pure Mathematics,
School of  Mathematical Scienes, 
University of Adelaide, 
Adelaide, SA 5005, 
Australia}

\email{mathai.varghese@adelaide.edu.au}

\begin{abstract}
Recently we initiated the study of spherical T-duality 
for spacetimes that are 
principal $\sfSU(2)$-bundles \cite{BEM14}. In this paper, we extend spherical T-duality 
to spacetimes that are oriented {\em non-principal} $\sfSU(2)$-bundles. There  are 
several interesting new examples in this case and a new phenomenon 
appearing in the non-principal case is the existence of infinitely many 
spherical T-duals.
\end{abstract}

\maketitle

% ===================================================================

\tableofcontents

\section{Introduction}
\label{sec:intro}

T-duality for pairs consisting of a circle bundle, together with a degree 3 H-flux, 
was originally studied in detail in \cite{BEM, BEM2, BHM, BHM05}
with contributions by several others later.  In string theory, T-dual pairs are distinct compactification 
manifolds that cannot be distinguished by any experiment, which is the notion of isomorphism relevant in physics.  
This equivalence in physics implies the isomorphisms of a number of other mathematical structures, such as Courant 
algebroids \cite{cavalcanti}, generalized complex structures \cite{cavalcanti} and twisted K-theory \cite{BEM}, see also 
\cite{BS,RR88}. It turns out that all of these structures are physically relevant.

Recently we initiated the study of spherical T-duality 
for principal $\sfSU(2)$-bundles in \cite{BEM14}. Let $P$ be a principal $\sfSU(2)$-bundle over $M$ and $H$ a 7-cocycle on $P$,
\begin{equation}\label{SU(2)1}
\begin{CD}
\sfSU(2) @>>> \,  P \\
&& @V \pi VV \\
&& M \end{CD}
\end{equation}

Principal $\sfSU(2)$-bundles over a compact oriented four dimensional manifold $M$ are classified 
by $\H^4(M;\ZZ) \cong \ZZ$ via the 2nd Chern class $c_2(P)$. This can be seen using the well known isomorphism,
$\H^4(M;\ZZ) \cong [M, S^4]\cong \ZZ$ and noting that there is a canonical principal $\sfSU(2)$-bundle 
$P\to S^4$, known as the Hopf bundle, whose 2nd Chern class is the generator of $\H^4(S^4;\ZZ)\cong \ZZ$.  
The orientation of $M$ and $\sfSU(2)$ imply that $\pi_*$ is a canonical isomorphism $\H^7(P;\ZZ)\cong\H^4(M;\ZZ)\cong\Z$.  
The dual principal $\sfSU(2)$-bundle,
\begin{equation}\label{SU(2)1}
\begin{CD}
\sfSU(2) @>>> \,  \widehat P \\
&& @V \widehat\pi VV \\
&& M \end{CD}
\end{equation}
is defined by $c_2(\widehat P)=\pi_*H$ while the dual 7-cocycle $\widehat H\in\H^7(\widehat P)$ satisfies $c_2(P) = 
{\widehat \pi}_*{\widehat H}$ by the isomorphism $\widehat\pi_* \colon \H^7(\widehat P;\ZZ)\cong\H^4(M;\ZZ)\cong\Z$.  
We proved that this spherical T-duality map induces degree-shifting isomorphisms between the real and integral twisted 
cohomologies of $P$ and $\widehat P$ and also between the 7-twisted K-theories.

Beyond dimension 4 the situation becomes more complicated as not all integral 4-cocycles of $M$ are realized as the 2nd
Chern class  of 
a principal $\sfSU(2)$-bundle $\pi:P\rightarrow M$, and multiple bundles can have the same $c_2(P)$.  We refer the 
reader to \cite{BEM14} for precise statements of spherical T-duality in the higher dimensional case. 
In string theory, spherical T-duality does not imply an isomorphism of two compactifications, but only of some of the 
data of these compactifications corresponding to certain conserved charges \cite{BEM14}.  The isomorphism of 
conserved charges is implied by the fact that spherical T-duality induces an isomorphism on twisted cohomology.

In this paper, we extend spherical T-duality 
to (oriented) non-principal $\sfSU(2)$-bundles. 
While principal $\sfSU(2)$-bundles correspond to unit sphere bundles
of quaternionic line bundles,  (oriented) non-principal $\sfSU(2)$-bundles
correspond to unit sphere bundles
of rank $4$ oriented real Riemannian vector bundles.
A striking new phenomenon in the non-principal case is 
that when the base $M$ is a compact oriented simply-connected 4 dimensional manifold, and given 
an oriented non-principal $\sfSU(2)$-bundle $E$ with $H$ a 7-cocycle on $E$,
then for each integer, we will show that there is an {\em infinite} lattice of spherical T-duals with 7-cocycle flux over $M$,
in stark contrast to the case of principal $\sfSU(2)$-bundles as described above.  
One reason is because in the non-principal
bundle case, the Euler class does not determine (oriented) non-principal $\sfSU(2)$-bundles when the base is as above, 
but in addition the 2nd Stiefel-Whitney
class and the Pontryagin class are also needed for the classification. However, it 
is only the Euler class (and its transgression) that is needed in the Gysin sequence and also to prove the isomorphisms
of 7-twisted integral cohomologies, and in addition the 2nd Stiefel-Whitney class is needed to prove the isomorphisms
of 7-twisted K-theories.  Since we no longer insist that the $\sfSU(2)$ bundles be principal, there are now a number of 
interesting new examples available, such as the spherical T-duality of Aloff-Wallach spaces with 7-cocycle flux described in Section 4. 
One relevant class of compactifications in string theory 
which has been of great interest over the past 30 years is that of Sasakian manifolds, 
 described in Section \ref{Sasakian manifolds}.\\

\noindent{\bf Acknowledgements.} The last author is grateful to Diarmuid Crowley for useful information 
regarding   $\sfSU(2)$-bundles. 
JE is supported by NSFC MianShang grant 11375201.
The research of PB and VM was supported under
Australian Research Council's Discovery Project funding scheme
(Project numbers DP110100072 and DP130103924).

% ===================================================================

\section{Oriented non-principal $\sfSU(2)$-bundles}
\label{sec:non-principal}

We begin with a few definitions. Recall that
\begin{definition}
\label{def:spherebundle}
A \emph{principal $\sfSU(2)$-bundle} $P$ over
a space $M$,
\begin{equation}\label{circle1}
\begin{CD}
\sfSU(2) @>>> \,  P\\
&& @V \pi VV \\
&& M \end{CD}
\end{equation}
is a map $P\xrightarrow{\pi} M$, where $P$ is a space
equipped with a free action of $\sfSU(2)$ 
for which $\pi$ is the projection onto the orbit space $M=P/\sfSU(2)$ and the
map is locally trivial, in the sense that for any $x\in X$, there exists
an open neighbourhood $U$ containing $x$ for which one has a commuting diagram
\[
\xymatrix{p^{-1}(U) \ar[rr]^{\cong} \ar[dr]^\pi && U \times \sfSU(2)
\ar[dl]_{\text{pr}_1}\\
& \,U &}
\]
An (oriented) \emph{non-principal $\sfSU(2)$-bundle} $E$ over $M$, on the other hand, is a
fiber bundle over $M$ with fiber $\sfSU(2)$ and structure group
$\Diff_+(\sfSU(2))$, the orientation preserving diffeomorphisms of $\sfSU(2)$. 
Since the inclusion $\sfSO(4)=\Isom_+(\sfSU(2))\longhookrightarrow 
\Diff_+(\sfSU(2))$ of the orientation preserving isometries of $\sfSU(2)$ into the 
orientation preserving diffeomorphisms of $\sfSU(2)$ is a
homotopy equivalence by Hatcher's theorem \cite{Hatcher}, it follows that 
the quotient space $\sfSO(4) \backslash\Diff_+(\sfSU(2))$ is contractible, therefore
{\em there is no loss of generality} in assuming that a general (oriented)
non-principal $\sfSU(2)$ bundle $E$ is a
fiber bundle over $M$ with fiber $\sfSU(2)$ and structure group
$\sfSO(4)$.
\end{definition}

Recall that $\sfSU(2)=\sfSp(1)$ is canonically identified with the group of unit quaternions, 
and consider the surjective morphism
$$
\sfSp(1) \times \sfSp(1) \longrightarrow \sfSO(4) 
$$
sending a pair of unit quaternions $(z,w)$ to the map $\RR^4 \rightarrow \RR^4$, $x \mapsto z x w$ 
using quaternion multiplication, where we identify $\RR^4$ with the quaternions $\mathbb{H}$. 
The kernel of this map consists of $(1, 1)$ and $(-1, -1)$ generating $\ZZ_2$, since $\RR$ is central in 
$\mathbb{H}$. Therefore $$\sfSO(4)  \cong (\sfSp(1) \times \sfSp(1))/\ZZ_2.$$
It follows that  $\pi_3(\sfSO(4)) = \ZZ \oplus \ZZ$. We can choose an explicit identification as follows: given 
${(p,q) \in \mathbb{Z}}$, we have a map ${\phi_{(p,q)}: \sfSp(1) \rightarrow \sfSO(4)}$ which sends a unit quaternion ${u}$ to the map
$$
 \phi_{(p,q)}(u): \mathbb{R}^4 \rightarrow \mathbb{R}^4, \ \phi_{(p,q)}(u)(x) = u^p x u^q\,. 
$$
Recall that principal $\sfSO(4)$-bundles $Q$ over $S^4$ are classified by $\pi_3(\sfSO(4)) = \ZZ \oplus \ZZ$. 
Let $(p, q) \in \pi_3(\sfSO(4))$ define a principal $\sfSO(4)$-bundle $Q(p,q)$ over $S^4$, by using the clutching
function $\phi_{(p,q)}$. By 
\cite{Milnor56}, the first Pontryagin class $p_1(Q(p,q)) = 2(p-q)$ and by \cite{Steen}, 
the Euler class $e(Q(p,q)) = p+q$. Therefore the  first Pontryagin class together with the Euler class 
determine all principal $\sfSO(4)$-bundles over $S^4$. An example of a principal $\sfSO(4)$-bundle 
over $S^4$ is the oriented frame bundle
$$
\begin{CD}
\sfSO(4)@>>> \,  \sfSO(5)\\
&& @V \pi VV \\
&& S^4
\end{CD}
$$
We calculate $(p,q)$ for this example. The Euler class of this bundle is just the Euler class
of $S^4$ and is equal to 2. So $p+q=2$. The Pontryagin class of this bundle is just the Pontryagin class of
of $S^4$, which is zero since $S^4$ is the oriented boundary of the 5 dimensional ball. So $2(p-q)=0$.
Therefore $p=1=q$, and $Q(1,1)=\sfSO(5)$.

We summarise the above, studied in \cite{CE}.

\begin{lemma}
\label{lemma:nonprincipal SU(2)}
Any principal $\sfSO(4)$-bundle $Q$ over $S^4$
is classified by the invariants $p_1(Q)$ and $e(Q)$ in 
$H^4(S^4;\ZZ) \cong \ZZ$.
\end{lemma}

Let $M$ be a compact, oriented, 4 dimensional manifold. 
Then for each pair of integers $(p,q)$, define a principal $\sfSO(4)$-bundle $Q_M(p, q)$ over $M$ as follows. For any degree one map $f: M \to S^4$, 
set $Q_M(p, q) = f^*(Q(p, q))$. Then the Euler number of $Q_M(p, q)$ is $e(Q_M(p, q)) = p+q \in \bZ \cong \H^4(M;\ZZ)$ and the 1st 
Pontryagin number of $Q_M(p, q)$ is 
$p_1(Q_M(p, q)) = 2(p-q) \in  \bZ \cong \H^4(M;\ZZ)$. \\

However these are not the only principal $\sfSO(4)$-bundles over $M$. We assume that $M$ is  simply-connected here.
A simple obstruction theory argument (cf. \cite{DW}) shows that a principal $\sfSO(4)$-bundle $Q$ over $M\setminus B(R, x_0)$,
where $B(R, x_0)$ is a small $4D$-ball of radius $R$ centered at $x_0 \in M$, is classified by the second Stiefel-Whitney class
$w_2(Q) \in \H^2(M\setminus B(R, x_0);\ZZ_2) \cong \H^2(M;\ZZ_2)$. Now the restriction of $Q$ to the boundary 
$\partial(M\setminus B(R, x_0))=S^3$ is determined 
by a map $h: S^2 \to \sfSO(4)$. But any such map is null homotopic since $\pi_2(\sfSO(4))\cong 0$, therefore the restriction of $Q$ to the boundary
$S^3$ is trivializable. 
Gluing the bundle over the $4D$ ball $B(R, x_0)$ corresponds to fixing trivialisations over the boundary $S^3$ and also 
prescribing a map $S^3\mapsto \sfSO(4)$. This fixes the Pontryagin class $p_1$ and the Euler class 
$e$ as shown earlier. Then the main Theorem in \cite{DW} (with useful clarifications in \cite{CE,EZ}) shows that \\

\begin{lemma}
\label{lemma:nonprincipal SU(2) on M}
Any principal $\sfSO(4)$-bundle $Q$ over a simply-connected, compact, oriented 4 dimensional manifold $M$ 
is classified by the invariants $w_2(Q) \in H^2(M; \ZZ_2)$,  $p_1(Q)$ and $e(Q)$ in 
$\H^4(M;\ZZ) \cong \ZZ$. There is no constraint on $w_2(Q)$. Let $b \in H^2(M; \ZZ)$ be such that
$w_2(Q)=b \mod 2$. Then $p_1(Q) = 2(p-q) + \beta$ where $b\cup b = \beta \in \ZZ \cong \H^4(M;\ZZ)$,
and $e(Q) = p+q$  for some $p, q \in \ZZ$.
Hence we can parametrize $Q= Q_M(\beta, p, q)$, where $\beta, p, q$ are integers satisfying the constraints above.
\end{lemma}

Let us examine the special case when $M = \bCP^2$, which is also studied in \cite{EZ}.  
A principal $\sfSO(4)$-bundle $Q$ over $\bCP^2\setminus B(R, x_0)$,
where $B(R, x_0)$ is a small open $4D$-ball of radius $R$ centered at $x_0 \in \bCP^2$, is classified by its
second Stiefel-Whitney class $w_2(Q) \in \H^2(\bCP^2\setminus B(R, x_0); \ZZ_2) \cong  \H^2(\bCP^2; \ZZ_2)\cong \ZZ_2$. 
Now $\bCP^2\setminus B(R, x_0)$ 
retracts onto $\bCP^1$, therefore $Q$ is determined by a smooth map $f: S^1 \to \sfSO(4)$ 
whose homotopy class is
$w_2(Q)$, which we will assume is the {\em nontrivial} class. The restriction of $Q$ to the boundary $\partial(\bCP^2\setminus B(R, x_0))=S^3$ is determined 
by a map $S^2 \mapsto \sfSO(4)$. But any such map is null homotopic since $\pi_2(\sfSO(4))\cong0$, 
therefore the restriction of $Q$ to the boundary
$S^3$ is trivializable. Upon fixing trivialisations over the boundary $S^3$ and also 
prescribing a map $S^3\mapsto \sfSO(4)$, fixes the Pontryagin class $p_1$ and the Euler class $e$. 
Such a map corresponds to a pair of integers $(p,q)$ and we denote  the resulting 
$\sfSO(4)$-bundle over $\bCP^2$ by $Q_{\bCP^2}(1,p,q)$. \\

\begin{lemma}
\label{lemma:nonprincipal SU(2) on CP2}
Any principal $\sfSO(4)$-bundle $Q$ over $\bCP^2$ 
is classified by the invariants $w_2(Q) \in \H^2(\bCP^2, \ZZ_2)\cong \ZZ_2$,  $p_1(Q)$ and $e(Q)$ in 
$\H^4(\bCP^2;\ZZ) \cong \ZZ$. More precisely, \\

\begin{enumerate}
\item if $w_2(Q)=0$, then in the notation of Lemma \ref{lemma:nonprincipal SU(2) on M},  for each pair of integers $(p,q)$, $Q \cong Q_{\bCP^2}(0, p, q)$ 
whose Euler number $e(Q_{\bCP^2}(0,p, q)) = p+q \in \bZ \cong \H^4(\bCP^2;\ZZ)$ and whose 1st Pontryagin number is 
$p_1(Q_{\bCP^2}(0,p, q)) = 2(p-q) \in  \bZ \cong \H^4(\bCP^2;\ZZ)$. \\

\item if $w_2(Q)\ne 0$, then  for each pair of integers $(p,q)$, $Q\cong Q_{\bCP^2}(1,p, q)$ 
as above, whose Euler number $e(Q_{\bCP^2}(1,p, q)) = p+q \in \bZ \cong \H^4(\bCP^2;\ZZ)$ and whose 1st Pontryagin number is 
$p_1(Q_{\bCP^2}(1,p, q)) = 2(p-q)+1 \in  \bZ \cong \H^4(\bCP^2;\ZZ)$. \\

\end{enumerate}

\end{lemma}

Consider the example of the oriented frame bundle,
$$
\begin{CD}
\sfSO(4)@>>> \,  \sfSO(\bCP^2)\\
&& @V \pi VV \\
&& \bCP^2
\end{CD}
$$
over $\bCP^2$. Since 
$\bCP^2$ is not a spin manifold, it follows that $w_2( \sfSO(\bCP^2)) \ne 0$. One calculates
the Euler number, $e(\bCP^2) = 3$ and the signature, $\text{sign}(\bCP^2) =1$. By the 
Atiyah-Singer index theorem, it follows that $p_1(\bCP^2)=3$. By Lemma \ref{lemma:nonprincipal SU(2) on CP2},
$3=p+q$ and $3=2(p-q)+1$. Therefore $p=2, q=1$ and $\sfSO(\bCP^2) \cong Q_{\bCP^2}(1,2, 1).$\\

We can use Lemma \ref{lemma:nonprincipal SU(2) on CP2} to make Lemma \ref{lemma:nonprincipal SU(2) on M} more explicit. 
Since $M$ is a simply-connected, compact, oriented 4 dimensional manifold, one 
uses the universal coefficient theorem and Poincar\'e duality to see that $\H^2(M;\ZZ)$ 
is torsion free and $\H^2(M;\ZZ_2)\cong (\ZZ_2)^{b_2(M)}$. Given a principal $\sfSO(4)$-bundle $Q$ over $M$
with $w_2(Q)\ne 0$,
then $w_2(Q) = b\mod 2$ for some $b\in \H^2(M;\ZZ) \cong  [M, \bCP^2]$. So $b:M \to \bCP^2$ and one 
can pullback the bundles in Lemma \ref{lemma:nonprincipal SU(2) on CP2} via $b$, and 
show that $Q= b^*(Q_{\bCP^2}(1,p,q))$, for some $p, q$ that are integers as in Lemma \ref{lemma:nonprincipal SU(2) on M}.
\\

When $\dim(M)>4$, the invariants in Lemma \ref{lemma:nonprincipal SU(2) on M} do not completely classify principal $\sfSO(4)$-bundles.
However, they do classify these bundles rationally.

%=============================
\section{Examples of oriented non-principal $\sfSU(2)$ bundles}
\label{sect:ex}

This section contains examples of oriented non-principal $\sfSU(2)$-bundles that are obtained from 
principal $\sfSO(4)$-bundles via the associated bundle construction.
Some of the examples can be found in \cite{Taubes}, Section 10.6.\\

\noindent {\bf Example 1.} For each pair of integers $(p,q)$, define a nonprincipal $\sfSU(2)$-bundle $E(p, q)$ over $S^4$ as the associated
bundle $E(p,q) = Q(p,q) \times_{\sfSO(4)}\sfSU(2)$, where $Q(p,q)$ is the principal $\sfSO(4)$-bundle 
constructed above Lemma \ref{lemma:nonprincipal SU(2)}. In particular, $E(1,1) = \sfSO(5) \times_{\sfSO(4)}\sfSU(2)$.
The Euler number of $E(p, q)$ is $e(E(p, q)) = p+q \in \bZ \cong H^4(S^4)$ and the 1st Pontryagin number of $E(p, q)$ is 
$p_1(E(p, q)) = 2(p-q) \in  \bZ \cong H^4(S^4)$ \cite{Milnor56}.

Explicitly, let $U_N$ 
be the open hemisphere in $S^4$ with the north pole as centre and $U_S$ 
be the open hemisphere in $S^4$ with the south pole as centre. Consider the trivial bundles $U_N \times \sfSU(2)$
and $U_S \times \sfSU(2)$ with transition function $\tilde g(x, h) = (x, g(x)^p h g(x)^q)$ for all 
$(x, h) \in (U_N \cap U_S) \times \sfSU(2)$. Here $g: U_N \cap U_S \to \sfSU(2)$ is a degree 1 map, where we 
observe that $U_N \cap U_S$ is a retraction to $S^3$. Let $E(p, q)$ be the (a priori) non-principal $\sfSU(2)$-bundle obtained via the 
clutching construction. Then $E(p, q)$ is a principal $\sfSU(2)$-bundle if and only if
either $p=0$ or $q=0$,  $E(p, q)$ is homeomorphic to $S^7$ if and only if $p+q=\pm1$. If $p+q=\pm1$ and $(p-q)^2-1\ne 0$
(mod 7), then $E(p, q)$ is not diffeomorphic to $S^7$, and in fact not mutually diffeomorphic. See \cite{CE} for 
details on the homeomorphism and diffeomorphism statements.
\\

\noindent {\bf Example 2.} Let $M$ be a compact, oriented, simply-connected, 4 dimensional manifold. 
Then for integers $\beta, p, q$ satisfying the constraints in Lemma \ref{lemma:nonprincipal SU(2) on M}, 
we have a principal $\sfSO(4)$-bundle $Q_M(\beta, p, q)$ over $M$. Define the nonprincipal $\sfSU(2)$-bundle 
$E_M(\beta, p, q)$ over $M$ to be the associated bundle 
$E_M(\beta, p,q) = Q_M(\beta, p,q) \times_{\sfSO(4)}\sfSU(2)$. Then the 2nd Stiefel-Whitney class, Euler class and Pontryagin class of 
$E_M(\beta, p,q)$ coincide with those of $Q_M(\beta, p,q)$ as in Lemma \ref{lemma:nonprincipal SU(2) on M}.\\

For examples that are possibly more directly relevant to physics, see Section \ref{Sasakian manifolds}.

%=============================
\section{Gysin sequence and construction of spherical T-duals for oriented non-principal 
$\sfSU(2)$-bundles }
\label{sec:gysin}

The goals in this section are to state the relevant Gysin sequence for (oriented) non-principal $\sfSU(2)$-bundles,
and then to construct spherical T-duals for oriented non-principal $\sfSU(2)$-bundles.

Given an (oriented) non-principal $\sfSU(2)$-bundle 
$E\xrightarrow{\pi} M$, it
comes with a Gysin sequence, Proposition 14.33, \cite{Bott-Tu}, 
\begin{equation}\label{gysin}
\cdots \to \H^p(M; \bZ) \xrightarrow{\cup e(E)} \H^{p+4}(M; \bZ)  
\xrightarrow{\pi^*} \H^{p+4}(E; \bZ) \xrightarrow{\pi_!}  
\H^{p+1}(M;\bZ) \xrightarrow{\cup e(E)} \cdots .
\end{equation}
where $e(E)$ denotes the Euler class of $E$.

For each pair of integers $(p,q)$, consider the nonprincipal $\sfSU(2)$-bundle $E(p, q)$ over $S^4$ 
with $w_2(E(p,q))=0$, $e(E(p,q))=p+q$ and $p_1(E(p,q))=2(p-q)$. 
Let $H = h\, \text{vol}$ be a 7-cocycle on $E(p,q)$
\begin{equation}\label{SU(2)1}
\begin{CD}
\sfSU(2) @>>> \,  E(p,q) \\
&& @V \pi VV \\
&& S^4 \end{CD}
\end{equation}
The orientation of $S^4$ and $\sfSU(2)$ together with the Gysin sequence in Eqn.~\eqref{gysin} imply that $\pi_*$ 
is a canonical isomorphism $\H^7(E(p,q);\ZZ)\cong\H^4(S^4;\ZZ)\cong\Z$.  Consider the family of spherical T-dual bundles 
\begin{equation}\label{SU(2)2}
\begin{CD}
\sfSU(2) @>>> \,  E(\hat p, \hat q) \\
&& @V \widehat\pi VV \\
&& S^4 \end{CD}
\end{equation}
with the property that $e(E(\hat p, \hat q))= \hat p + \hat q = h$ while the dual 7-cocycle $\widehat H = \hat h \,\text{vol} 
\in\H^7(E(\hat p, \hat q);\ZZ)$ satisfies $e(E(p,q)) = 
p+q = \hat h$ by the isomorphism $\widehat\pi_* \colon \H^7(E(\hat p, \hat q) ;\ZZ)\cong\H^4(S^4;\ZZ)\cong\Z$.  \\

{\em 
Thus we see that for each integer $\hat p$, there is a spherical T-dual pair $(E(\hat p, h-\hat p), (p+q) \,\mathrm{vol})$ over $S^4$
to any fixed spherical pair $(E(p, q), h\, \mathrm{vol})$ over $S^4$.\\
}

More generally, 
let $M$ be a compact, oriented, simply-connected, 4 dimensional manifold and
consider for each triple of integers $(\beta, p,q)$, consider the nonprincipal $\sfSU(2)$-bundle 
\begin{equation}\label{SU(2)M}
\begin{CD}
\sfSU(2) @>>> \,  E_M(\beta,p,q) \\
&& @V \pi VV \\
&& M \end{CD}
\end{equation}
with $w_2(E_M(\beta,p,q)) = b\mod 2\neq 0$, $e(E_M(\beta,p,q))=p+q$ and $p_1(E_M(\beta,p,q))=2(p-q) +\beta$, where $b, \beta$ are 
as in Lemma \ref{lemma:nonprincipal SU(2) on M} and Example 2 in Section \ref{sect:ex}.  Arguing as above, the family of spherical T-dual bundles 
\begin{equation}\label{SU(2)2}
\begin{CD}
\sfSU(2) @>>> \,  E_M(\beta,\hat p, \hat q) \\
&& @V \widehat\pi VV \\
&& M \end{CD}
\end{equation}
with the property that $e(E_M(\beta,\hat p, \hat q))= \hat p + \hat q = h$, $p_1(E_M(\beta,\hat p, \hat q))=2( \hat p-\hat q) +\beta$, 
$w_2(E_M(\beta,\hat p, \hat q)) = w_2(E_M(\beta,p,q))$ while the dual 7-cocycle $\widehat H = \hat h \,
\text{vol} \in\H^7(E_M(\beta,\hat p, \hat q);\ZZ)$ 
satisfies $e(E_M(\beta,\hat p, \hat q)) = 
p+q = \hat h$.\\

{\em 
Thus we see that for each integer $\hat p$, there is a spherical T-dual pair $(E_M(\beta,\hat p, h-\hat p), (p+q) \,\mathrm{vol})$ over $M$
to any fixed spherical pair $(E_M(\beta,p, q), h\, \mathrm{vol})$ over $M$.\\
}

In contrast, recall from \cite{BEM14} that in the case of spherical pairs $(P(q), h)$ where $P(q)$ is a {\em principal} $\sfSU(2)$-bundle over 
a compact, oriented 4 dimensional manifold $M$ with 2nd Chern class $c_2(P(q))=q\in \ZZ \cong \H^4(M; \ZZ)$ and 
$h \in \ZZ \cong \H^7(P(q);\ZZ)$ an integral 7-cocycle on $P(q)$, there is a {\em unique} spherical T-dual pair $(P(h), q)$.
The reason for the family of spherical T-duals in the case of
non-principal $\sfSU(2)$-bundles over $M$ is that a
non-principal $\sfSU(2)$-bundle is no longer  determined by just 
its Euler class (and its second Stiefel-Whitney class). 
In fact, as we have seen, there is a lattice of non-principal $\sfSU(2)$-bundles over $M$ 
having the same second Stiefel-Whitney class $w_2$ and Euler class $e$. 

\begin{example}
The Aloff-Wallach space $W_{k,l}$ \cite{AW}, 
defined as being the homogeneous space $\sfSU(3)/T_{k,l}$,
where the circle subgroup $T_{k,l} = \mathrm{diag}(z^k, z^l, z^{-(k+l)})$, $|z|=1$, is a non-principal 
$S^3$-bundle over $\bCP^2$ iff $|k+l|=1$.  We have 
$W_{p,1-p} \cong S_{-1,p(p-1)}$, where $S_{p,q} = E_{\bCP^2}(1,p, q)$, in our notation above,
and $e(S_{a,b}) = a-b$, $p_1(S_{a,b})=2(a+b)+1$ \cite{EZ}.  
Hence
\begin{equation*}
e(W_{p,1-p} ) = - (p^2 - p +1) \,,\qquad p_1(W_{p,1-p} ) = 2p(p-1)-1\,, \qquad 
w_2(W_{p,1-p} ) =1 \,.
\end{equation*}
Note that for $(k,l)=(p,1-p)$:
\begin{equation*}
k^2 + l^2 + kl = p^2 + (p-1)^2 - p(p-1) = p^2 -p +1 \,,
\end{equation*}
consistent with $H^4(W_{k,l},\ZZ) \cong Z_{|k^2 + l^2 + kl |}$ \cite{AW}.
Note that $(W_{p,1-p}, h)$ is dual to any $S_{\hat p,\hat q}$ with $h=\hat p - \hat q$.
In particular, to be dual to an Aloff-Wallach space $W_{\hat p , 1 -\hat p}$ we need
$\hat p (\hat p -1 ) - (1 + h ) = 0$, i.e.\
\begin{equation*}
\hat p = \frac12 \left( 1 \pm \sqrt{ 1 - 4 (1+h) } \right) \,.
\end{equation*}
After parametrizing these values of $h$ by integers $\hat p$, we  find the duality
between $(W_{p,1-p}, h= -(\hat p^2 -\hat p +1) )$ and $(W_{\hat p,1-\hat p}, \hat h= -(p^2 - p +1) )$.
This gives examples of cases where 
3-Sasakian manifolds are spherical T-dual to other 3-Sasakian manifolds. See also the last 
section for a further discussion of 3-Sasakian manifolds.
\end{example}

%=============================
\section{Isomorphism of integral 7-twisted cohomologies and  K-theories for spherical 
T-dual pairs}
\label{sec:T-duality isom}

When $M$ is a compact, oriented, simply-connecetd, 4 dimensional manifold, we will prove in this section 
that the spherical T-duality map induces parity changing isomorphisms between the integral 7-twisted cohomologies of 
$(E_M(\beta,p, q), h\, \text{vol})$ and $(E_M(\beta,\hat p, h-\hat p), (p+q) \,\text{vol})$ for each integer $\hat p$.

We begin by considering the case of $S^4$ and the non-principal $\sfSU(2)$-bundle  $E(p,q)$ over it. The 7-twisted cohomology 
$\H^{even/odd}(E(p,q), h\, \text{vol};\ZZ)$ is defined as the cohomology of the $\ZZ_2$-graded complex $\left(\H^{even/odd}(E(p,q); \ZZ),  
h \text{vol}\right)$. Using the Gysin sequence
in Eqn.~\eqref{gysin} to calculate the cohomology groups
$\H^{even/odd}(E(p,q); \ZZ)$, we obtain for $p+q\ne 0$
\begin{align*}
\H^j(E(p,q); \ZZ) & \cong 0, \, j\ne 0,4,7 \,,\\
\H^4(E(p,q); \ZZ) & \cong \ZZ_{p+q} \,, \\
\H^7(E(p,q); \ZZ) & \cong\ZZ \cong \H^0(E(p,q); \ZZ)\,.
\end{align*}
It follows that for $p+q\ne 0, h\ne 0$ one has
\begin{align*}
\H^{even}(E(p,q), h \text{vol}; \ZZ) & \cong  \ZZ_{p+q} \,,\\
\H^{odd}(E(p,q), h \text{vol}; \ZZ) & \cong \ZZ_h\,.
\end{align*}
Therefore, by our explicit computation above, for each integer $\hat p$, there exists an isomorphism of 7-twisted cohomology 
groups over the integers with a parity change,
\begin{align*}
\H^{even}(E(p,q), h \text{vol}; \ZZ) &\cong \H^{odd} (E(\hat p, h-\hat p), (p+q) \,\text{vol}; \ZZ)\,,\\
\H^{odd}(E(p,q), h \text{vol}; \ZZ) &\cong \H^{even} (E(\hat p, h-\hat p), (p+q) \,\text{vol}; \ZZ)\,.
\end{align*}
A choice of transgression of the Euler class gives an explicit isomorphism, as outlined in the next section.\\

More generally, let $M$ be a simply-connected compact, oriented, 4 dimensional manifold, and consider
the non-principal $\sfSU(2)$-bundle $E_M(\beta,p,q)$ over $M$ as in Lemma \ref{lemma:nonprincipal SU(2) on M} and Example 2 in Section \ref{sect:ex},
together with the 7-cocycle  $H = h \text{vol}$ on $E_M(\beta,p,q)$. 
Again, use the Gysin sequence of Eqn.~\eqref{gysin} to calculate the cohomology groups
$\H^{even/odd}(E_M(\beta,p,q); \ZZ)$, we obtain for $p+q\ne 0$
\begin{align*}
\H^j(E_M(\beta,p,q); \ZZ) & \cong \H^{4-j}(M; \ZZ), \, j= 0,1,2,3 \,,\\
\H^4(E_M(\beta,p,q); \ZZ) & \cong \ZZ_{p+q} \oplus \H^1(M;\ZZ) \,,\\
\H^{7-j}(E_M(\beta,p,q); \ZZ) & \cong \H^{4-j}(M; \ZZ), \, j= 0,1,2,3 \,.
\end{align*}
It follows that for $p+q\ne 0, h\ne 0$ one has
\begin{align*}
\H^{even}(E_M(\beta,p,q), h \text{vol}; \ZZ) & \cong \ZZ_{p+q}  \oplus \H^2(M;\ZZ) \oplus  \H^1(M;\ZZ)\oplus  \H^3(M;\ZZ) \,,\\
\H^{odd}(E_M(\beta,p,q), h \text{vol}; \ZZ) & \cong \ZZ_h \oplus \H^2(M;\ZZ) \oplus \H^1(M;\ZZ) \oplus  \H^3(M;\ZZ)\,.
\end{align*}
So it follows by our explicit computation above that for each integer $\hat p$, there is an isomorphism of 7-twisted cohomology groups over the integers with a parity change,
\begin{align*}
\H^{even}(E_M(\beta,p,q), h \text{vol}; \ZZ) &\cong \H^{odd} (E_M(\beta,\hat p, h-\hat p), (p+q) \,\text{vol}; \ZZ),\\
\H^{odd}(E_M(\beta,p,q), h \text{vol}; \ZZ) &\cong \H^{even} (E_M(\beta,\hat p, h-\hat p), (p+q) \,\text{vol}; \ZZ).\\
\end{align*}

The proof of the isomorphism, up to an extension problem, between 7-twisted K-theory and twisted cohomology proceeds 
similarly to Section 6.1.2 in \cite{BEM14} in the case of spherical pairs $(E_M(\beta,p,q), h \text{vol})$, 
\begin{align*}
\H^{even}(E_M(\beta,p,q), h \text{vol}; \ZZ) &\cong \K^0 (E_M(\beta,p,q), h \text{vol})\,,\\
\H^{odd}(E_M(\beta,p,q), h \text{vol}; \ZZ) &\cong \K^1(E_M(\beta,p,q), h \text{vol})\,.
\end{align*}
Therefore we conclude that there is an isomorphism of 7-twisted K-theories for spherical T-dual pairs with a parity change,
\begin{align*}
\K^0(E_M(\beta,p,q), h \text{vol}) &\cong \K^1(E_M(\beta,\hat p, h-\hat p), (p+q) \,\text{vol})\,,\\
\K^1(E_M(\beta,p,q), h \text{vol}) &\cong \K^0 (E_M(\beta,\hat p, h-\hat p), (p+q) \,\text{vol})\,.\\
\end{align*}

We now give more details. For brevity we will set $E=E_M(\beta,p,q)$, the coefficient rings will always be $\Z$.  
The 7-twisted K-theory can be constructed using a two step spectral sequence with differentials $d_1=Sq^3$ and 
$d_2=h \text{vol}\cup$. The first differential acts trivially on $\H^0(E)$, $\H^1(E)$ and $\H^2(E)$ as $Sq^3:\H^k(E)\rightarrow\H^{k+3}(E)$ 
annihilates classes of dimension less than 3.  Similarly it annihilates $\H^5(E)$, $\H^6(E)$ and $\H^7(E)$ because there are no nontrivial 
classes of dimension greater than 7.  The image of $Sq^3\H^4(E)$ is a $\Z_2$-torsion element of $\H^7(E)$ but $E$ is oriented so 
there are no such nontrivial elements.  

In \cite{BEM14} it was shown that $Sq^3\H^3(E)$ is trivial using the Gysin sequence for principal bundles and 
applying the K\"unneth theorem to the case in which the characteristic class $c_2$ vanishes.  As the 
$S^3$-bundle is no longer necessarily principal in the present note, the characteristic class which appears in the Gysin 
sequence is the Euler class $p+q$ and when it vanishes the bundle is not necessarily trivial, therefore the K\"unneth 
theorem may not be applied.  However in the present case we have assumed that $M$ and therefore $E$ is simply-connected.  
By Poincar\'e duality $\H^6(E)=0$ and so again $Sq^3\H^3(E)=0$.  Therefore the first differential $d_1=Sq^3$ acts trivially on 
$\H^\bullet(E)$.  The cohomology of the second differential $d_2=h \text{vol}\cup$ then is, up to extension problem, isomorphic to the 
7-twisted K-theory.  Therefore, for brevity restricting our attention to the nontrivial case $h\neq 0$, the 7-twisted K-groups are
\begin{eqnarray}
K^0(E, h)&\cong&\H^2(E)\oplus\H^4(E)\oplus\H^6(E)\cong \H^2(M)\oplus \Z_{p+q}\oplus \H^3(M)\nonumber\\
K^1(E,h)&\cong&\H^1(E)\oplus\H^3(E)\oplus\H^5(E)\oplus \Z_h\cong \H^3(M)\oplus\H^2(M)\oplus \Z_h.
\end{eqnarray}
Therefore spherical T-duality interchanges $p+q$ and $h$ and so exchanges $K^0$ and $K^1$ as claimed.

%=============================
\section{T-duality isomorphisms via differential geometry}
\label{sec:non-principal}

\subsection{Euler class and Pfaffian form} Let $S$ be a nonprincipal $S^3$-bundle over a 
simply-connected compact, oriented, 4 dimensional manifold $M$. Then $S=S(E)$ is the unit sphere bundle 
of rank $4$ oriented real Riemannian vector bundle $E$ over $M$. Let $\nabla$ be a connection on $E$ preserving the 
Riemannian metric on $E$, having curvature $\Omega_\nabla$. Since $\Omega_\nabla$ is a 2-form on 
$M$ with values in skew-symmetric endomorphisms, one can consider the Pfaffian, $\text{Pfaff}(\Omega_\nabla)$, 
which is a degree 4 differential form on $M$. Then the Euler class $e(E) = e(S) \in \H^4(M;\ZZ)$ is represented by 
$\frac{1}{4\pi^2}\text{Pfaff}(\Omega_\nabla)$.\\

\subsection{Transgression of the Euler class} The pullback of the Euler class $\pi^*(e(S))=0$ where $\pi:S\to M$ is the 
projection. The transgression of the Euler class is a degree 3 differential form $\tau(\nabla)$ on $S$ such that 
$\pi^*(\text{Pfaff}(\Omega_\nabla)) = d \tau(\nabla)$. S.S.~Chern  was the first to define such a form, although 
the explicit expression that he obtained is somewhat complicated  \cite{C1, C2}. An alternate expression for $\tau(\nabla)$ was obtained in 
Section 7 of \cite{MQ} using the Gaussian shaped
Chern-Weil representative of the Thom class defined there using superconnections. See also the discussion in Section  3.3 of \cite{Zhang}. 

In more detail, let $\cA_t = \frac{t^2}{2} |x|^2 + t \nabla x - \pi^*(\Omega_\nabla)$ be defined on $E$, where $x$ is the fibre variable.   Let $\displaystyle \int^B$ 
denote the Berezin (or Fermionic) integral. When restricted to the  sphere bundle $S$, one has 
\begin{align*}
 \pi^*(\text{Pfaff}(\Omega_\nabla)) &=d \tau(\nabla)
 \tau(\nabla)  = \int_0^\infty dt \int^B \left(x\, e^{-\cA_t}\right).
 \end{align*}
 Also $\displaystyle \int_{S^3}  \tau(\nabla) =1$, 
 so the transgression form $\tau(\nabla)$ is the analog of the Chern-Simons form used in \cite{BEM14}.

 \subsection{T-duality isomorphisms via differential geometry}
 Given a pair of spherical T-dual pairs $(S, H)$, and $(\widehat S, \widehat H)$, we can decompose them as
 \begin{align*}
H &= \pi^* H_7 + \pi^*\text{Pfaff}(\Omega_{\widehat \nabla} ) \wedge \tau(\nabla)\,,\\
\widehat H &= \widehat\pi^* H_7 + \widehat\pi^* \text{Pfaff}(\Omega_{\nabla} ) \wedge \tau(\widehat\nabla)\,.
 \end{align*}
 Here $\widehat \nabla$ is the connection on the rank $4$ oriented real Riemannian vector bundle $\widehat E$ over $M$
 such that $\widehat S$ is the unit sphere bundle of $\widehat E$.
One checks easily that $dH = 0 = d\widehat H$.

Let $\omega$ be a $d_H$-closed form representing a class in $\H_H^{\rm{even/odd}}(S)$. Lifting to the correspondence space 
$S\times_{M} \widehat S$, applying the kernel $\exp(\tau(\nabla)\wedge \tau(\widehat \nabla))$, and integrating over the fiber, we define the 
{\em T-duality transform}
\begin{equation}
T_*(\omega) = \int_{\sfSU(2)} \exp(\tau(\nabla)\wedge \tau(\widehat \nabla)) \wedge \widehat p^*\omega,
\end{equation}
which one checks is a $d_{\widehat H}$-closed form. Here
\begin{equation*}
\xymatrix{
& \ar[dl]_{\widehat p} S\times_{M} \widehat S \ar[dr]^{p}&\\
(S,H)&&
(\widehat S, \widehat H)
}
\end{equation*}

\begin{theorem} \label{homlem}
The T-duality transform $T_*$ induces an isomorphism of 7-twisted cohomology groups
\begin{equation}
T_* : \H_H^{\rm{even/odd}}(S) \stackrel{\cong}{\longrightarrow} \H_{\widehat H}^{\rm{odd/even}}(\widehat S)\,.
\end{equation}
\end{theorem}
\begin{proof}
Since $d(\tau(A)\wedge \tau(\widehat A))=-\widehat p^*H+ p^*{\widehat H}$, we have
\bea
T_*(d_H\omega)&=&\int_{\sfSU(2)} \exp(\tau(A)\wedge \tau(\widehat A)) \wedge \widehat p^*d\omega\nonumber -\int_{\sfSU(2)}
\exp(\tau(A)\wedge \tau(\widehat A)) \wedge \widehat p^*H\wedge\widehat p^*\omega\nonumber\\
&=&-d\int_{\sfSU(2)} \exp(\tau(A)\wedge \tau(\widehat A)) \wedge \widehat p^*\omega\nonumber\\
&&-\int_{\sfSU(2)} \exp(\tau(A)\wedge \tau(\widehat A)) \wedge (\widehat p^*H-\widehat p^*H+ p^*{\widehat H})\wedge\widehat p^*\omega\nonumber\\
&=&-d_{\widehat{H}}T_*(\omega) \,, \label{tcomm}
\eea
where in the last step we used the fact that $\int_{\sfSU(2)}= p_*$, together with the adjunction property of the pullback
\begin{equation}
p_*(\alpha\wedge p^*\beta)=(p_*(\alpha))\wedge\beta\,.
\end{equation}
Eqn.~(\ref{tcomm}) may be summarized by the statement $T\circ d_{H}=-d_{{\widehat H}}\circ T$.
Therefore $T$ takes $d_H$-exact (closed) forms on $P$ to $d_{\widehat H}$-exact (closed) forms on $\widehat P$ 
and so it induces a well-defined homomorphism on the twisted cohomology groups.  
On verifies exactly as in Section 5 of \cite{BEM14} that it is an isomorphism.
\end{proof}

%=============================
\section{Examples: Sasakian Manifolds}
\label{Sasakian manifolds}

In the present context of non-principal bundles, again we have proven an isomorphism of twisted cohomology and so again 
spherical T-duality is guaranteed to provide an isomorphism of certain conserved charges in string theory.  However, 
as we now no longer insist that the sphere bundles be principal, there are now a number of new examples available.

One class of compactifications which has been of great interest over the past 30 years is that of Sasakian manifolds.  
More concretely, type IIB supergravity comes with a 10-dimensional compactification manifold which may be 
$AdS_5\times M^5$ and 11-dimensional supergravity comes with an 11-dimensional manifold which may be 
$AdS_4\times M^7$.  Here $AdS$ is a Lorentzian hyperbolic space.  If $M$ is a Sasaki-Einstein or 3-Sasakian manifold, 
then the corresponding compactification preserves some supersymmetry, for example in the case of a 7-manifold $M^7$ 
it will preserve $\mathcal{N}=2$ supersymmetry in the Sasaki-Einstein (and $\mathcal N=3$ in the 3-Sasakian) case.  

The examples of Sasaki-Einstein and 3-Sasakian manifolds \cite{CEZ} which have received the most attention in the physics 
literature are essentially all $S^3$ bundles.  The oldest and simplest is the 7-sphere $S^7$, which is a principal 
$S^3$-bundle over $S^4$ with $e=$vol.  Already in 5-dimensions the conifold $M^5=T^{1,1}$ is usually described 
as an $S^1$-bundle over the product $S_a^2\times S_b^2$ of 2-copies of the 2-sphere.  The Chern class is equal to 
the generator $a\in\H^2(S^2_a)$.  Combining the fiber with $S^2_a$ one obtains a trivial $S^3$ fibration over $S^2_b$.  
The slightly more complicated examples $Y^{p,q}$ \cite{GMSW04} are circle bundles over $S^2_a\times S^2_b$ with Chern class 
$pa+qb$ where $a$ and $b$ generate $\H^2(S^2_a)$ and $\H^2(S^2_b)$ respectively.  When $p$ and $q$ are relatively 
prime again these are homeomorphic to $S^3\times S^2$ and so provide trivial $S^3$-bundles.

The 7-dimensional case is nontrivial in the current context because in these compactifications the volume of $M^7$ is a 
monotonic function of the 7-flux integrated over $M^7$ and in particular the $M^7$ only has nonvanishing volume when 
the 7-flux, which we have called $h$vol but is traditionally called $*G_4$, is nonvanishing.  Therefore, even if the 7-dimensional 
$S^3$-bundle is trivial, the spherical T-dual will be nontrivial.  

The two most popular Sasaki-Einstein examples of $M^7$ are $M^{1,1,1}$ \cite{Wi81} and $Q^{1,1,1}$ \cite{DFvN84}.  
The first is a circle bundle over $\bCP^1\times\bCP^2$ with Chern class $2a+3b$ where $a$ and $b$ are the 
generators of $\H^2(\bCP^1)$ and $\H^2(\bCP^2)$ respectively.  The second is a circle bundle over 3 copies 
of $\bCP^1$ where the Chern class is the sum of the generators of $\H^2(\bCP^1)$ of the three copies.  In the case 
case of $Q^{1,1,1}$, the circle fiber can be combined with any one of the $\bCP^1$'s on the base to 
create an $S^3$ fibered over the remain $\bCP^1\times\bCP^1$.  The fibration will have a trivial Euler class.  
Nonetheless, as $h\neq 0$ in these compactifications, the spherical T-duals will necessarily have nontrivial 
Euler classes.  As a result it seems likely that the spherical T-duals do not admit a Sasaki-Einstein metric and 
so will not manifest as much supersymmetry as the original compactifications.  Nonetheless an isomorphism of 
those conserved charges classified by twisted cohomology is guaranteed by the isomorphism of the twisted cohomology.

%=================================================================

%=================================================================

\end{document}